\documentclass[11pt]{article}

\usepackage[utf8]{inputenc}
\usepackage[T1]{fontenc}
\usepackage[english]{babel}
\usepackage{amsmath,amssymb,amsthm}
\usepackage{booktabs}
\usepackage{graphicx}
\usepackage{multirow}
\usepackage{array}
\usepackage[margin=1in]{geometry}
\usepackage{tikz}
\usepackage{pgfplots}
\usepackage{algorithm}
\usepackage{algpseudocode}
\usepackage{caption}
\usepackage{subcaption}
\usepackage[numbers,sort&compress]{natbib}
\usepackage[colorlinks=true,linkcolor=blue,citecolor=blue,urlcolor=blue]{hyperref}
\usepackage{cleveref}

\pgfplotsset{compat=1.17}
\usetikzlibrary{positioning,arrows.meta,calc,fit,backgrounds,shapes.geometric,decorations.pathreplacing}

\newtheorem{proposition}{Proposition}
\newtheorem{remark}{Remark}

\DeclareMathOperator{\CV}{CV}
\DeclareMathOperator*{\argmin}{arg\,min}
\newcommand{\R}{\mathbb{R}}
\newcommand{\Ecos}{\mathcal{E}}

\tikzset{
  stage/.style={draw, rounded corners=2pt, align=center, font=\small,
                minimum height=8mm, inner xsep=4pt, fill=black!3},
  gate/.style={draw, rounded corners=2pt, align=center, font=\small,
               minimum height=8mm, fill=black!8},
  det/.style={draw, rounded corners=2pt, align=center, font=\small,
              minimum height=9mm, fill=black!5},
  flow/.style={-{Latex[length=2mm]}, thick},
  lbl/.style={font=\scriptsize\itshape}
}

\title{\bfseries ZAPs: A Reward Attribution Framework for DeFi Ecosystems with
Adversarial-Robust Scoring via Parallel Anomaly Ensemble Detection}

\author{%
\begin{tabular}{@{}ccc@{}}
Girish G N & Ashutosh Sahoo & Ajay Bhat \\
{\small\itshape Head of AI/ML} &
{\small\itshape Chief Executive Officer} &
{\small\itshape Chief Operating Officer} \\
{\small Zeru Finance} &
{\small Zeru Finance} &
{\small Zeru Finance} \\
\noalign{\vskip 4mm}
Akshay SP & Gurukiran S & Parag Paul \\
{\small\itshape Chief Product Officer} &
{\small\itshape Chief Technology Officer} &
{\small\itshape Advisor} \\
{\small Zeru Finance} &
{\small Zeru Finance} &
{\small Zeru Finance} \\
\noalign{\vskip 5mm}
\multicolumn{3}{c}{Dhanashekar Kandaswamy} \\
\multicolumn{3}{c}{\small\itshape The Ohio State University, USA}
\end{tabular}%
}

\date{}

\begin{document}

\maketitle

\begin{abstract}
Incentive programs dominate user acquisition in decentralized finance. Points
systems, retroactive airdrops, and liquidity-mining campaigns distribute
billions in token value against on-chain activity, yet nearly all attribute
rewards through naive heuristics---raw volume, transaction count, wallet
count---that bots and sybil operations manufacture at industrial scale. The
result is a systematic transfer from genuine participants to industrialized
farming. This paper presents ZAPs, an attribution framework that addresses the
failure jointly rather than piecewise. A composite activity score normalizes
user volume against each protocol's own population distribution via
high-percentile anchoring, bounding whale dominance while preserving
differentiation across the remaining population. A two-layer cross-domain
weighting mechanism---protocol-share within sector, multiplied by sector-share
within ecosystem---prices protocol significance and makes niche-protocol
farming structurally unprofitable; we prove that the maximum extractable reward
at any protocol is bounded by that protocol's global volume share, independent
of how cheaply local dominance is purchased. A four-layer adversarial detection
stack combines transaction-level gates, a parallel anomaly ensemble pairing a
one-class reconstruction model against an isolation-forest outlier detector,
post-distribution memory, and graph-based sybil clustering, all feeding
graduated, non-binary penalty multipliers. On a labeled corpus of 1{,}073
malicious wallets spanning 124{,}638 transactions, the ensemble reaches
$0.923 \pm 0.013$ ROC-AUC against $0.891 \pm 0.016$ for the reconstruction model
alone---but only when the isolation forest is fit on the benign population;
fitting it on the pooled population inverts its polarity and destroys the
ensemble gain entirely. Controlled adversarial simulation reduces reward capture
by 30--90\% across attack archetypes while legitimate-user scenarios shift by
1--8\%, and live production campaigns recorded a 56\% reduction in sybil
allocation, a 49\% increase in quality-wallet participation, and a 50\%
reduction in sell pressure. This work extends the zScore reputation lineage
\citep{zscore,deepreputation} from behavioral reputation to reward attribution,
where direct financial incentives make adversarial robustness a first-order
design requirement.
\end{abstract}

\section{Introduction}
\label{sec:intro}

Decentralized finance has largely replaced the loyalty and rewards apparatus of
traditional financial platforms with permissionless incentive design. Where a
centralized exchange grants fee tiers to verified customers, a DeFi protocol
distributes tokens to pseudonymous addresses based solely on observed on-chain
behavior. The scale of the substitution is substantial: retroactive airdrops and
points programs now routinely allocate token supplies worth hundreds of millions
of dollars, and eligibility is computed, not negotiated \citep{airdropwhy}.

The attribution mechanisms behind these distributions have not kept pace with
their economic weight. Most programs reward raw activity---total volume
transacted, number of transactions, number of participating wallets---under the
implicit assumption that activity proxies for genuine participation. That
assumption fails in an adversarial setting, and the setting is now reliably
adversarial. Farming operations manufacture qualifying activity at industrial
scale: spam bots generate thousands of low-value transactions, sybil operators
split capital across hundreds of coordinated wallets \citep{sybil}, and volume
routes through obscure protocols where small absolute amounts translate into
large relative shares. Each strategy converts a naive attribution heuristic
directly into extractable value.

Prior work, including the zScore lineage \citep{zscore,deepreputation}, has
addressed the adjacent problem of wallet \textit{reputation}---scoring the
quality and consistency of on-chain behavior for credit assessment,
segmentation, and trust modeling. Reputation scoring is not attacked the way
reward attribution is. A reputation score misjudging a wallet produces a
mislabeled profile; a reward function misjudging a wallet pays the attacker. The
direct financial payoff changes the threat model qualitatively.

This paper asks the question directly: can a reward-attribution mechanism
simultaneously reward genuine economic contribution in proportion to ecosystem
significance, and remain robust to adversarial manufacturing of that same
activity? We answer affirmatively and present ZAPs, a deployed framework whose
design and evaluation constitute the following contributions.

\begin{itemize}

\item \textbf{A composite reward score} (\Cref{sec:composite}) combining
percentile-normalized volume, engagement duration, and an extensible behavioral
quality signal, bounded so that no component dominates the composite through
scale alone.

\item \textbf{A two-layer cross-domain weighting mechanism}
(\Cref{sec:fairness}) establishing a formal notion of ecosystem-relative
economic significance. \Cref{prop:niche} proves that maximum extractable
attribution at a protocol is bounded by that protocol's global volume share,
which closes the most common structural exploit in volume-based programs.

\item \textbf{A four-layer adversarial detection stack}
(\Cref{sec:adversarial}) combining per-trade integrity gates, a parallel anomaly
ensemble, post-distribution memory, and graph-based cross-wallet sybil
clustering, feeding graduated penalty multipliers that bound false-positive harm
relative to binary exclusion.

\item \textbf{Empirical validation} (\Cref{sec:results,sec:ablation}) combining
deployment statistics, a held-out evaluation on 1{,}073 labeled malicious
wallets, controlled adversarial simulation, and live production campaigns. We
identify a training-regime condition on the isolation forest that is
outcome-determining for the ensemble and, to our knowledge, unreported: fitting
the detector on the pooled population rather than the benign population inverts
its discriminative polarity.

\end{itemize}

\section{Related Work}
\label{sec:related}

\subsection{Reputation and Scoring Systems in Crypto}

Wallet-level reputation modeling developed rapidly as DeFi shifted analytical
attention from platforms to participants. Udupi et al.\ introduced zScore, a
universal cross-protocol reputation system built on semi-supervised neural
scoring of on-chain activity \citep{zscore}. The same lineage was specialized to
decentralized exchanges, combining an interpretable behavioral blueprint with a
noise-augmented deep residual network to score liquidity providers and traders
\citep{deepreputation}. Lin et al.\ proposed RiskProp, propagating
de-anonymization-derived risk ratings across the Ethereum transaction graph
\citep{riskprop}, and Nguyen et al.\ applied PageRank to borrower
trustworthiness in lending protocols \citep{pagerep}. Jain et al.\ combined
on-chain behavioral metrics with machine learning in a general-purpose Web3
reputation engine \citep{wire}, while Packin and Lev-Aretz examined the opacity
and fairness risks such systems inherit \citep{blackbox}.

These systems answer the question \textit{how trustworthy or sophisticated is
this wallet?} The present work answers a different one: \textit{how much should
this wallet be paid?}---a question in which the scoring function itself becomes
the attack surface. ZAPs is positioned as the attribution-layer continuation of
the zScore lineage, inheriting its behavioral-modeling foundations while adding
the fairness weighting and adversarial machinery that direct financial
distribution requires.

\subsection{Market Structure and Volume Semantics}

The economics of decentralized exchange establish why raw volume is a poor
proxy for contribution. Malamud and Rostek showed that decentralized market
structure alters welfare properties relative to centralized venues
\citep{malamud}; Lo and Medda documented the emergence of automated market
makers as a distinct microstructure \citep{lomedda}; and Aspris et al.\
characterized the resulting trading environment and its weak surveillance
guarantees \citep{aspris}. Work on liquidity provision quantified the gap
between nominal volume and realized economic return
\citep{heimbach,aigner}. Volume transacted and value contributed diverge, and
an attribution mechanism that conflates them inherits the divergence.

\subsection{Incentive Design and Sybil-Resistant Distribution}

The mechanism-design literature offers principled treatments of fair allocation
under strategic behavior, most prominently quadratic funding, which reasons
explicitly about the tension between capital-weighted and participant-weighted
allocation \citep{quadratic}---the same tension any reward program must resolve.
Allen et al.\ analyze airdrops from an economic perspective, documenting their
role as ownership-distribution and marketing instruments \citep{airdropwhy}.

Academic treatment of \textit{adversarially robust} airdrop and points design is
thin. Practitioner literature abounds---post-mortems of farmed distributions, ad
hoc sybil-hunting bounties---but formal frameworks connecting attribution
formulas to attack economics are scarce. The gap is itself a finding: the
mechanisms distributing the largest sums in the ecosystem have received the
least formal scrutiny.

\subsection{Sybil Detection in Decentralized Systems}

Douceur's foundational result establishes that sybil attacks cannot be prevented
in permissionless systems without external identity anchors---only made
expensive \citep{sybil}. Social-graph approaches such as SybilGuard bound sybil
influence via trust-network topology \citep{sybilguard}, and Victor developed
address-clustering heuristics resolving multiple Ethereum addresses to common
controlling entities through deposit-address reuse and airdrop-claim patterns
\citep{clustering}. The clustering layer of \Cref{sec:sybil} extends this
tradition to reward attribution: first-funder provenance graphs and behavioral
uniformity scoring within candidate families detect coordinated multi-wallet
operations that per-wallet detectors cannot reach.

\subsection{Anomaly Detection Methods}

Isolation Forest isolates anomalies through random recursive partitioning,
exploiting the property that outliers require fewer splits to isolate
\citep{isolationforest}; it remains a strong baseline for unlabeled tabular
anomaly detection and serves as one of two detectors in our ensemble.
Autoencoder-based detection takes the complementary approach: a network trained
to reconstruct a reference distribution reconstructs out-of-distribution inputs
poorly \citep{autoencoder1,autoencoder2}. We invert the usual polarity, training
on malicious wallets so that \textit{low} reconstruction error indicates
membership in the malicious manifold. Chalapathy and Chawla survey the broader
deep anomaly-detection landscape \citep{deepsurvey}. In credit-scoring contexts,
Dastile et al.\ \citep{creditsurvey} and B\"{u}cker et al.\
\citep{credittransparency} emphasize the interpretability obligations of ML
scoring systems---a constraint addressed here through per-feature
reconstruction-error attribution.

\subsection{Positioning}

The individual ingredients of ZAPs---percentile normalization, anomaly
ensembles, multiplicative weighting, graph-based clustering---are each
established. The contribution is their joint composition into a
reward-attribution mechanism: nested cross-domain fairness weighting closing
structural formula-level exploits, and a four-layer detection stack closing
behavioral activity-level exploits, evaluated together against both attack
archetypes and legitimate-user controls.

\section{Methodology}
\label{sec:method}

\Cref{fig:pipeline} gives the end-to-end architecture. The remainder of this
section develops each stage.

\begin{figure}[t]
\centering
\begin{tikzpicture}[x=1mm,y=1mm]

\node[stage,text width=27mm] (raw)  at (0,0)  {Multi-chain\\transaction stream};
\node[gate, text width=38mm] (filt) at (40,0)
  {Ingestion gates\\{\scriptsize zero-value $\cdot$ approve/revoke $\cdot$ ceiling}};
\node[stage,text width=24mm] (pair) at (80,0) {$(w,p)$ pair\\decomposition};
\draw[flow] (raw) -- (filt);
\draw[flow] (filt) -- (pair);

\node[stage,text width=22mm] (vol)  at (-2,-26) {Volume\\{\scriptsize percentile}};
\node[stage,text width=22mm] (hold) at (23,-26) {Engagement\\{\scriptsize capped span}};
\node[stage,text width=18mm] (qual) at (46,-26) {Quality\\{\scriptsize zScore}};
\begin{scope}[on background layer]
\node[draw,rounded corners=2pt,fit=(vol)(hold)(qual),inner sep=3mm,fill=black!2]
  (comp) {};
\end{scope}
\node[font=\scriptsize\bfseries,anchor=south west] at (comp.north west)
  {Composite score $s_{w,p}$};

\node[stage,text width=52mm] (weight) at (22,-48)
  {Cross-domain weight\\[1mm]
   $\omega_p=\dfrac{V_p}{V_{\sigma(p)}}\cdot\dfrac{V_{\sigma(p)}}{V_{\Ecos}}
   =\dfrac{V_p}{V_{\Ecos}}$};

\node[det,text width=20mm] (l1) at (78,-20) {L1\\{\scriptsize Trade gates}};
\node[det,text width=20mm] (l2) at (78,-32) {L2\\{\scriptsize Anomaly ens.}};
\node[det,text width=20mm] (l3) at (78,-44) {L3\\{\scriptsize Dump memory}};
\node[det,text width=20mm] (l4) at (78,-56) {L4\\{\scriptsize Sybil cluster}};
\begin{scope}[on background layer]
\node[draw,dashed,rounded corners=2pt,fit=(l1)(l2)(l3)(l4),inner sep=2.6mm]
  (stack) {};
\end{scope}
\node[font=\scriptsize\bfseries,anchor=south,yshift=2.4mm] at (stack.north)
  {Adversarial stack};

\node[gate,text width=44mm] (pen) at (78,-70)
  {Graduated penalty $\mu(w)\in(0,1]$};
\node[stage,fill=black!10,text width=56mm] (out) at (34,-84)
  {$R_w=\mu(w)\displaystyle\sum_{p\in\mathcal{P}_w}s_{w,p}\,\omega_p$};

\draw[flow] (pair.south) -- ++(0,-5) -| (comp.north);
\draw[flow] (comp.south) -- (weight.north);
\draw[flow] (pair.east) -- ++(6,0) |- (stack.east);
\draw[flow] (stack.south) -- (pen.north);
\draw[flow] (pen.south) |- (out.east);
\draw[flow] (weight.west) -- ++(-12,0) |- (out.west);

\end{tikzpicture}
\caption{ZAPs attribution pipeline. The structural path (composite score,
cross-domain weighting) and the behavioral path (the four-layer adversarial
stack) run in parallel and compose multiplicatively at the attribution step.}
\label{fig:pipeline}
\end{figure}
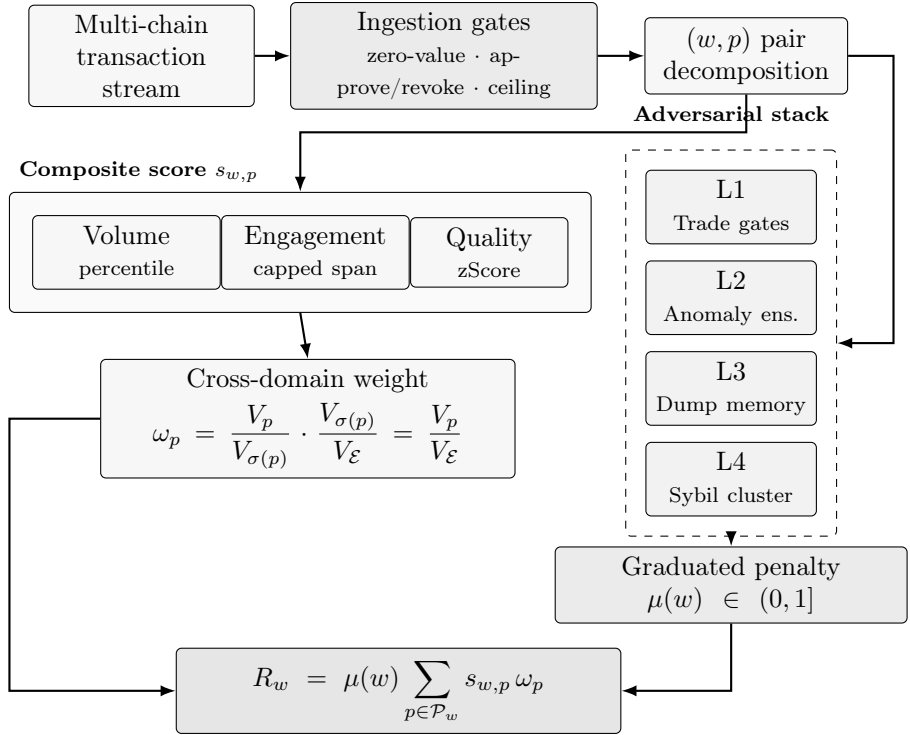

\subsection{Data Collection and Preprocessing}
\label{sec:data}

The pipeline draws from five DeFi sectors---DEX swaps, lending, perpetual
futures, NFT trading, and staking---each carrying transaction semantics that a
single aggregate volume figure would erase before attribution begins. A swap, a
borrow, and a stake represent different economic commitments, and the framework
treats them as such. Protocols are assigned to sectors through a maintained
classification; protocols absent from the mapping are excluded rather than
defaulted, on the principle that an uncategorized protocol should not silently
receive credit.

Three filters encode the system's definition of economic signal. Zero-value
transactions are discarded: they commit no capital. Approval and revocation
operations are excluded, since they modify permissions rather than positions.
Transactions exceeding a normalization ceiling are dropped, because at that
scale records are overwhelmingly artifacts of pricing error or internal
bookkeeping, and their inclusion would distort the population distribution used
to score every other participant.

Attribution is computed per (wallet, protocol) pair, not per wallet. Let
$\mathcal{W}$ denote the wallet set, $\mathcal{P}$ the protocol set, and
$\sigma:\mathcal{P}\to\mathcal{S}$ the protocol-to-sector assignment. A user
active across several protocols receives independent scores at each, summed into
a single reward. The decomposition is load-bearing for the fairness mechanism of
\Cref{sec:fairness}: a wallet cannot farm a niche protocol and carry that credit
into a dominant sector. The weight is applied where the activity occurred.

\subsection{Composite Reward Score}
\label{sec:composite}

Each pair $(w,p)$ receives a score composed from three independent dimensions:
activity volume, temporal engagement, and behavioral quality. Writing $S$ for
the common per-dimension scale,
\begin{equation}
s_{w,p} \;=\; \alpha\, v_{w,p} \;+\; \beta\, h_{w,p} \;+\; \gamma\, q_{w,p},
\qquad \alpha+\beta+\gamma = 1,\quad \alpha,\beta,\gamma > 0,
\label{eq:composite}
\end{equation}
so that $s_{w,p}\in[0,S]$. The dimensions combine additively---none substitutes
for the others---and each is individually bounded, so no single component
dominates the composite through scale alone. The mixing coefficients
$(\alpha,\beta,\gamma)$ are calibrated empirically and withheld; the analysis
that follows depends only on their positivity and on the boundedness of each
term.

\paragraph{Volume.} Let $V_{w,p}$ be the wallet's dollar-denominated activity at
protocol $p$, and let $\Pi_p^{(\tau)}$ denote the $\tau$-th percentile of the
volume distribution over all participants at that protocol, for a high fixed
$\tau$. Then
\begin{equation}
v_{w,p} \;=\; S \cdot \min\!\left(\frac{V_{w,p}}{\Pi_p^{(\tau)}},\, 1\right).
\label{eq:volume}
\end{equation}
The reference is population-derived and per-protocol: a venue dominated by
institutional participants develops a high reference, one of small retail users
a low one, and identical dollar volumes score relative to the protocol's own
economic scale. The choice of a high percentile over the maximum or the mean is
deliberate. Max-normalization makes the reference a function of a single
outlier, so one anomalous wallet compresses every other score toward zero.
Mean-normalization places the reference inside the bulk of the distribution, so
the clip in \Cref{eq:volume} binds for a large fraction of participants and
differentiation collapses. A high percentile is robust to individual outliers
while remaining far enough into the tail that saturation affects only genuine
whales. \Cref{sec:ablation} quantifies both failure modes.

\paragraph{Engagement.} Let $\Delta t_{w,p}$ be the span between the wallet's
first and last interaction with $p$, and $T^{\ast}$ a design-chosen ceiling:
\begin{equation}
h_{w,p} \;=\; S \cdot \min\!\left(\frac{\Delta t_{w,p}}{T^{\ast}},\, 1\right).
\label{eq:hold}
\end{equation}
The ceiling screens out single-transaction tourists without imposing indefinite
lock-in. A wallet that appears once and vanishes contributes nothing here;
sustained participation saturates the cap.

\paragraph{Quality.} $q_{w,p}$ is an extensible slot for sector-specific
behavioral indicators---repayment discipline in lending, realized-PnL character
in perpetuals, token-quality mix in DEX activity---backed by the model of
\Cref{sec:behavioral}. In the reference deployment this component is held at a
calibrated baseline for all users, which provides a small floor of credit
independent of volume. \Cref{sec:limitations} treats the consequences.

\subsection{Cross-Domain Fairness Weighting}
\label{sec:fairness}

\Cref{eq:composite} answers a local question: how substantial was this activity
relative to this protocol? It cannot answer a global one: how significant is
this protocol to the ecosystem? Any mechanism that skips the second question is
structurally exploitable, because an attacker's optimal strategy becomes
locating the smallest protocol whose scoring reference can be dominated cheaply.

Let $V_p$ denote total volume at protocol $p$, $V_{\sigma(p)}=\sum_{p'
:\sigma(p')=\sigma(p)} V_{p'}$ total volume in its sector, and $V_{\Ecos}$ total
ecosystem volume. The effective weight is the product of two nested shares:
\begin{equation}
\omega_p
\;=\;
\underbrace{\frac{V_p}{V_{\sigma(p)}}}_{\text{protocol within sector}}
\cdot
\underbrace{\frac{V_{\sigma(p)}}{V_{\Ecos}}}_{\text{sector within ecosystem}}
\;=\;
\frac{V_p}{V_{\Ecos}}.
\label{eq:weight}
\end{equation}

The product telescopes to the protocol's global share, but the decomposition is
retained deliberately. A two-layer structure makes the fairness arithmetic
auditable at each level, allows a deployment to reweight a sector without
recomputing intra-sector proportions, and ensures the classification step and
the weighting step share a single reference point.

Total attribution for a wallet, over its active protocol set $\mathcal{P}_w$, is
\begin{equation}
R_w \;=\; \mu(w) \sum_{p \in \mathcal{P}_w} s_{w,p}\,\omega_p,
\label{eq:attribution}
\end{equation}
where $\mu(w) \in (0,1]$ is the penalty multiplier of
\Cref{sec:adversarial}.

\begin{proposition}[Niche-protocol unprofitability]
\label{prop:niche}
Fix any protocol $p$. The attribution any single wallet can obtain at $p$
satisfies
\[
s_{w,p}\,\omega_p \;\le\; S\cdot\frac{V_p}{V_{\Ecos}},
\]
independently of $V_{w,p}$, of the wallet's share of $p$, and of the cost of
acquiring that share.
\end{proposition}

\begin{proof}
Each of $v_{w,p}, h_{w,p}, q_{w,p}$ lies in $[0,S]$ by
\Cref{eq:volume,eq:hold} and the construction of $q$. Since
$\alpha+\beta+\gamma=1$ with all coefficients positive, \Cref{eq:composite}
gives $s_{w,p}\le S$. Substituting the telescoped form of \Cref{eq:weight}
yields $s_{w,p}\,\omega_p \le S\,V_p/V_{\Ecos}$.
\end{proof}

\begin{remark}
The bound is the structural content of the mechanism. An attacker who
manufactures \emph{total} dominance of a protocol still cannot extract more
than that protocol's global volume share, so the return to capturing a venue is
governed by the venue's ecosystem weight rather than by how cheaply it can be
captured. Saturation in \Cref{eq:volume} makes the bound tight: once
$V_{w,p}\ge\Pi_p^{(\tau)}$, additional capital yields no additional volume
credit.
\end{remark}

The effect is sharpest numerically. Consider two protocols with comparable local
dominance in the reference deployment. Protocol~A holds a $19.4\%$ share of the
DEX sector, which carries $85.7\%$ of global volume, giving
$\omega_A \approx 0.166$. Protocol~B holds a $71.3\%$ share of the staking
sector---near-total local dominance---but staking carries $0.47\%$ of global
volume, giving $\omega_B \approx 0.0034$. Protocol~B's users dominate their
sector roughly three times more thoroughly and receive about one-fiftieth of the
weight.

\subsection{Deep Behavioral Quality Model}
\label{sec:behavioral}

The quality signal is backed by zScore, the behavioral engine behind Zeru's
reputation systems \citep{zscore,deepreputation}. The same engine that scores
wallet trustworthiness for credit assessment feeds the quality dimension of
ZAPs, ensuring that attribution and reputation share a behavioral substrate
rather than drifting into incompatible representations of the same wallet.

The model evaluates on-chain history across five dimensions---wealth magnitude,
behavioral consistency, protocol diversity, token-category diversity, and
gas-consumption behavior---each normalized against population-derived references
recomputed from periodic sampling of the live wallet base. A wallet's score
reflects its position relative to the current user distribution, not a static
benchmark that drifts into irrelevance as the population evolves.

Two modeling decisions bear directly on adversarial properties.

\paragraph{Frequency-adjusted consistency.} Raw flow volatility penalizes
high-frequency participants automatically: a market-making wallet executing many
transactions daily exhibits large absolute fluctuations as a mechanical
consequence of cadence, not as evidence of erratic behavior. Writing
$\bar{\varsigma}_w$ for the mean of the wallet's inbound and outbound flow
standard deviations and $\rho_w$ for its transaction rate, the model scores
\begin{equation}
\varsigma^{\text{eff}}_w \;=\; \frac{\bar{\varsigma}_w}{1+\ln(1+\rho_w)},
\qquad
c_w \;=\; c_{\min} + (C-c_{\min})\cdot
\mathrm{clip}\!\left(1-\frac{\varsigma^{\text{eff}}_w}{\Theta},\,0,\,1\right),
\label{eq:consistency}
\end{equation}
for a reference dispersion $\Theta$ and a guaranteed floor $c_{\min}$. Dividing
by a log-compressed transaction rate measures regularity conditional on the
wallet's natural operating tempo. A capital-intensity gate caps the consistency
contribution for wallets below a minimum lifetime value, which defeats the
complementary exploit of manufacturing behavioral regularity through dust
transactions.

\paragraph{Economic-substance weighting of diversity.} Protocol and token
diversity are not raw counts. Each token is assigned a tier weight
$\theta(t)$ reflecting economic substance---stablecoins, majors, and
liquid-staking derivatives at the top; established governance and blue-chip DeFi
assets next; unclassified assets at unit weight; and assets matching
points-, test-, and airdrop-voucher patterns discounted by more than an order of
magnitude relative to the top tier. The diversity input is
$\sum_{t \in \mathcal{T}_w}\theta(t)$ rather than $|\mathcal{T}_w|$, passed
through a saturating transform. An attacker who manufactures diversity by
deploying disposable contracts and minting valueless tokens generates a
quantity the model treats as noise.

A deterministic persona-tagging layer sits downstream of scoring. A stable hash
of the wallet's score bucket and identifier selects among candidate labels,
guaranteeing that repeated scoring of the same wallet produces identical
categorical output on any worker in a distributed deployment, without
coordination or shared state.

\subsection{Adversarial Detection}
\label{sec:adversarial}

The structural defenses of \Cref{sec:composite,sec:fairness} close formula-level
exploits. They cannot distinguish a human trader from a bot executing identical
nominal volume. ZAPs layers four defense classes against behavioral
manufacturing.

\subsubsection{Layer 1: Transaction-level gates}

Before any behavioral model fires, each transaction passes checks for signals
with no honest explanation. Same-block round-trip trades---a position opened and
closed within a single block, carrying zero net economic exposure---contribute
nothing to any score, irrespective of nominal magnitude. Fee-to-volume ratios
below a protocol's plausible range are treated identically, indicating rebate or
wash-trading mechanics rather than genuine capital commitment. These checks are
binary, deterministic, and applied upstream of everything else.

\subsubsection{Layer 2: Parallel anomaly ensemble}
\label{sec:ensemble}

Two detectors operate over a common wallet-level feature space
$x \in \R^{d}$, $d=10$, aggregating each wallet's full transaction history into
transaction count, active-day count, mean and standard deviation of transaction
value, total value, mean gas cost, distinct counterparty count, distinct token
count, token-transaction share, and mean daily transaction rate.

\paragraph{Detector A: one-class reconstruction model.}
Let $f_\theta = g \circ e$ with encoder $e:\R^{d}\to\R^{k}$ and decoder
$g:\R^{k}\to\R^{d}$, realized as $10 \to 64 \to 16 \to 64 \to 10$ with ReLU
activations and a linear output layer (\Cref{fig:ae}). Training minimizes
reconstruction error over a corpus of \emph{labeled malicious wallets only},
\begin{equation}
\theta^{\ast} = \argmin_{\theta}\; \frac{1}{|\mathcal{D}_{\text{mal}}|}
\sum_{x \in \mathcal{D}_{\text{mal}}} \big\| x - f_\theta(x) \big\|_2^2 ,
\label{eq:ae}
\end{equation}
by Adam with learning rate $10^{-3}$, batch size $128$, for $50$ epochs, with
inputs standardized to zero mean and unit variance and best-validation
checkpointing on a held-out $15\%$ split.

The polarity is inverted relative to standard autoencoder anomaly detection.
Because the model sees only malicious behavior, \emph{low} reconstruction error
indicates membership in the malicious manifold. Writing
$\varepsilon(x)=\tfrac{1}{d}\|x-f_{\theta^\ast}(x)\|_2^2$ and
$(\mu_\varepsilon,s_\varepsilon)$ for its mean and standard deviation over the
training corpus, the emitted probability is
\begin{equation}
P_{\mathrm{A}}(x) \;=\; \Big(1+\exp\big(z(x)\big)\Big)^{-1},
\qquad
z(x) = \frac{\varepsilon(x)-\mu_\varepsilon}{s_\varepsilon},
\label{eq:aeprob}
\end{equation}
which is decreasing in reconstruction error and bounded in $(0,1)$. Because
$\varepsilon$ decomposes per coordinate as
$\varepsilon_j(x)=(x_j-f_{\theta^\ast}(x)_j)^2$, the model emits a ranked
attribution over the top-$k$ contributing features, satisfying the
interpretability obligations emphasized in the credit-scoring literature
\citep{creditsurvey,credittransparency}.

\begin{remark}[On the latent dimension]
\label{rem:overcomplete}
With $d=10$ inputs and $k=16$ latent units the representation is
\emph{over-complete}: the architecture imposes no dimensional compression and
could in principle realize the identity map. The regularization that makes
\Cref{eq:aeprob} discriminative comes from one-class training---the model is
never shown benign wallets---together with ReLU sparsity and early stopping,
not from a capacity bottleneck. We flag this because the design is easily
misread as compressive, and because it predicts the failure mode observed in
\Cref{sec:ablation}: the detector generalizes to unseen malicious wallets but
carries no explicit pressure to reconstruct benign wallets poorly.
\end{remark}

\begin{figure}[t]
\centering
\begin{tikzpicture}[x=1cm,y=0.40cm,
  unit/.style={circle,draw,minimum size=2.6mm,inner sep=0pt}]

\foreach \i in {1,...,5}
  \node[unit] (in\i)  at (0,{4-\i*1.6}) {};
\foreach \i in {1,...,7}
  \node[unit,fill=black!5]  (h1\i) at (2,{6.4-\i*1.6}) {};
\foreach \i in {1,...,5}
  \node[unit,fill=black!15] (b\i)  at (4,{4-\i*1.6}) {};
\foreach \i in {1,...,7}
  \node[unit,fill=black!5]  (h2\i) at (6,{6.4-\i*1.6}) {};
\foreach \i in {1,...,5}
  \node[unit] (out\i) at (8,{4-\i*1.6}) {};

\foreach \i in {1,...,5}\foreach \j in {1,...,7}
  \draw[gray!30,line width=0.15pt] (in\i) -- (h1\j);
\foreach \i in {1,...,7}\foreach \j in {1,...,5}
  \draw[gray!30,line width=0.15pt] (h1\i) -- (b\j);
\foreach \i in {1,...,5}\foreach \j in {1,...,7}
  \draw[gray!30,line width=0.15pt] (b\i) -- (h2\j);
\foreach \i in {1,...,7}\foreach \j in {1,...,5}
  \draw[gray!30,line width=0.15pt] (h2\i) -- (out\j);

\node[font=\scriptsize,align=center] at (0,-6.6) {input\\$d=10$};
\node[font=\scriptsize,align=center] at (2,-6.6) {ReLU\\$64$};
\node[font=\scriptsize,align=center] at (4,-6.6) {latent\\$k=16$};
\node[font=\scriptsize,align=center] at (6,-6.6) {ReLU\\$64$};
\node[font=\scriptsize,align=center] at (8,-6.6) {output\\$d=10$};

\draw[decorate,decoration={brace,amplitude=4pt},thick]
  (-0.3,5.4) -- (3.9,5.4) node[midway,above=3pt,font=\scriptsize] {encoder $e$};
\draw[decorate,decoration={brace,amplitude=4pt},thick]
  (4.1,5.4) -- (8.3,5.4) node[midway,above=3pt,font=\scriptsize] {decoder $g$};

\node[align=left,font=\scriptsize,anchor=north west] at (8.9,3.2)
  {trained on \textbf{malicious only}\\[3pt]
   low $\varepsilon(x)\Rightarrow$ malicious\\[3pt]
   $k>d$: over-complete\\(\Cref{rem:overcomplete})};

\end{tikzpicture}
\caption{One-class reconstruction detector. The latent width exceeds the input
width, so discrimination derives from one-class training rather than
compression.}
\label{fig:ae}
\end{figure}
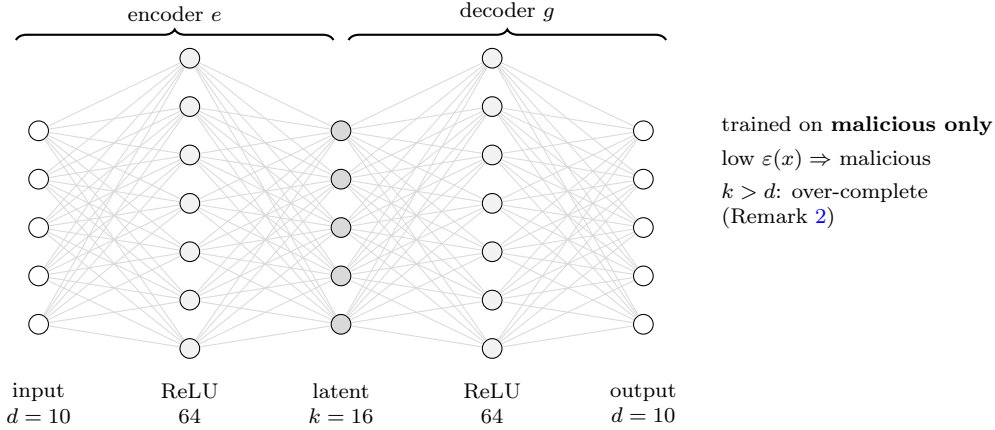

\paragraph{Detector B: isolation forest.}
An isolation forest \citep{isolationforest} operates on the same feature space,
scoring points by the average path length $E[h(x)]$ required to isolate them
under random recursive partitioning, normalized by the expected path length
$c(n)$ of an unsuccessful binary-search over $n$ points:
$\mathrm{iso}(x) = 2^{-E[h(x)]/c(n)}$. Trained without labels, it captures
statistical outlier-ness independent of any previously observed archetype. We
map its raw outlier score $r(x)$ to a comparable scale by
\begin{equation}
P_{\mathrm{B}}(x) \;=\; \Big(1+\exp\big(-(r(x)-\bar{r})/s_r\big)\Big)^{-1}.
\label{eq:ifprob}
\end{equation}
\Cref{sec:ablation} establishes that the population on which this detector is
fit determines whether \Cref{eq:ifprob} is informative or actively misleading.

\paragraph{Fusion.} The two probabilities combine convexly,
\begin{equation}
P(x) \;=\; \lambda\,P_{\mathrm{A}}(x) \;+\; (1-\lambda)\,P_{\mathrm{B}}(x),
\qquad \lambda \in [0,1].
\label{eq:fusion}
\end{equation}
The production value of $\lambda$ is calibrated against held-out adversarial and
legitimate-user scenarios and is withheld. \Cref{sec:ablation} reports the full
sweep over $\lambda$ from our own replication.

The detectors fail in complementary directions. A model trained on known attacks
misses genuinely novel automation that reconstructs unremarkably; an isolation
forest misses bots tuned to sit inside the population's typical parameter
ranges. An attacker must simultaneously avoid resembling every known malicious
archetype \emph{and} avoid statistical distinguishability from organic behavior.

\subsubsection{Layer 3: Post-distribution memory}

A wallet that liquidates the majority of received rewards within a short window
of distribution carries that record into subsequent allocation cycles. Future
eligibility is adjusted; rewards are not retroactively clawed back. For an
adversary operating at scale, every liquidation trades short-term extraction
against degraded positioning in every future allocation window referencing the
same behavioral record, which converts single-campaign optimization into a
multi-campaign constraint.

\subsubsection{Layer 4: Cross-wallet sybil clustering}
\label{sec:sybil}

Per-wallet detectors are blind to relationships between wallets. An operation
splitting activity across many addresses, each individually unremarkable, passes
per-wallet scrutiny entirely. The clustering layer operates on the funding
provenance graph (\Cref{fig:sybil}).

Let $\phi(w)$ denote the address that first funded wallet $w$. Candidate
families are the fibers $F_a = \{w : \phi(w) = a\}$ with $|F_a| \ge 2$. High-degree
funders corresponding to exchange and bridge infrastructure are excluded before
grouping; in the validation scan below, $39$ such funders were removed. Within
each surviving family, behavioral uniformity is scored by the mean coefficient
of variation across transaction-signature attributes $\mathcal{A}$ (cadence,
value distribution, protocol mix):
\begin{equation}
\CV(F) \;=\; \frac{1}{|\mathcal{A}|}\sum_{a \in \mathcal{A}}
\frac{\varsigma_a(F)}{\bar{m}_a(F)},
\label{eq:cv}
\end{equation}
where $\varsigma_a$ and $\bar{m}_a$ are the standard deviation and mean of
attribute $a$ across the family's wallets. Low $\CV$ indicates that
independently-operated wallets would not plausibly exhibit such uniformity;
shared funding provenance combined with low $\CV$ yields high-confidence sybil
classification. A family-level risk score composes $\CV(F)$, $|F_a|$, and a
funder-validation gate over four structural metrics.
\Cref{alg:sybil} states the procedure.

\begin{algorithm}[t]
\caption{Sybil family resolution over the funding-provenance graph}
\label{alg:sybil}
\begin{algorithmic}[1]
\Require wallet set $\mathcal{W}$; first-funder map $\phi$; degree bound
$D_{\max}$; uniformity bound $\CV_{\max}$; attribute set $\mathcal{A}$
\Ensure sybil families with risk scores
\State $\mathcal{F} \gets \emptyset$
\State $A \gets \{\phi(w) : w \in \mathcal{W}\}$
  \Comment{candidate funders}
\For{$a \in A$}
  \State $F_a \gets \{w \in \mathcal{W} : \phi(w) = a\}$
  \If{$|F_a| > D_{\max}$}
    \State \textbf{continue}
      \Comment{exchange or bridge hub; excluded before grouping}
  \EndIf
  \If{$|F_a| \ge 2$}
    \State $\mathcal{F} \gets \mathcal{F} \cup \{F_a\}$
  \EndIf
\EndFor
\For{$F \in \mathcal{F}$}
  \State $\CV(F) \gets \frac{1}{|\mathcal{A}|}\sum_{a \in \mathcal{A}}
         \varsigma_a(F) / \bar{m}_a(F)$
    \Comment{\Cref{eq:cv}}
  \State $g(F) \gets$ number of structural metrics passed by funder of $F$
  \State $\mathrm{risk}(F) \gets \Psi\big(\CV(F),\, |F|,\, g(F)\big)$
    \Comment{$\Psi$ decreasing in $\CV$, increasing in $|F|$ and $g$}
  \State \textbf{label} $F$ as sybil if $\CV(F) \le \CV_{\max}$ and $g(F)$
         meets the validation gate
\EndFor
\State \Return $\{(F, \mathrm{risk}(F)) : F \in \mathcal{F}\}$
\end{algorithmic}
\end{algorithm}

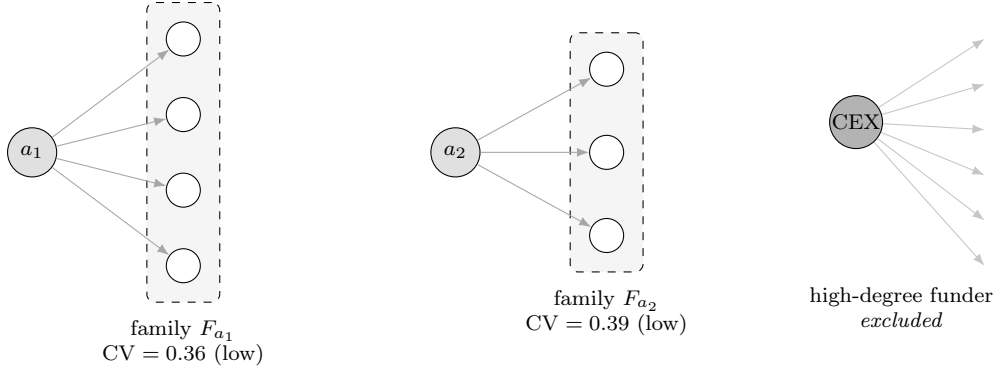
\begin{figure}[t]
\centering
\begin{tikzpicture}[
  fnd/.style={draw,circle,minimum size=6.5mm,inner sep=0pt,font=\scriptsize,fill=black!12},
  wal/.style={draw,circle,minimum size=4.5mm,inner sep=0pt,font=\tiny,fill=white},
  exch/.style={draw,circle,minimum size=7mm,inner sep=0pt,font=\scriptsize,fill=black!30},
  e/.style={-{Latex[length=1.6mm]},gray!70}
]

\node[fnd] (f1) at (0,0) {$a_1$};
\foreach \i/\y in {1/1.5, 2/0.5, 3/-0.5, 4/-1.5}
  \node[wal] (w1\i) at (2,\y) {};
\foreach \i in {1,...,4} \draw[e] (f1) -- (w1\i);
\begin{scope}[on background layer]
\node[draw,dashed,rounded corners=3pt,fit=(w11)(w14),inner sep=2.5mm,fill=black!4] (fam1) {};
\end{scope}
\node[font=\scriptsize,align=center,below=1mm of fam1]
  {family $F_{a_1}$\\$\CV=0.36$ (low)};

\node[fnd] (f2) at (5.6,0) {$a_2$};
\foreach \i/\y in {1/1.1, 2/0.0, 3/-1.1}
  \node[wal] (w2\i) at (7.6,\y) {};
\foreach \i in {1,...,3} \draw[e] (f2) -- (w2\i);
\begin{scope}[on background layer]
\node[draw,dashed,rounded corners=3pt,fit=(w21)(w23),inner sep=2.5mm,fill=black!4] (fam2) {};
\end{scope}
\node[font=\scriptsize,align=center,below=1mm of fam2]
  {family $F_{a_2}$\\$\CV=0.39$ (low)};

\node[exch] (ex) at (10.9,0.4) {CEX};
\foreach \y in {1.5,0.9,0.3,-0.3,-0.9,-1.5}
  \draw[e,gray!45] (ex) -- (12.6,\y);
\node[font=\scriptsize,align=center] at (11.5,-2.05) {high-degree funder\\\emph{excluded}};

\end{tikzpicture}
\caption{Funding-provenance clustering. Wallets sharing a first funder form
candidate families; behavioral uniformity within a family
(\Cref{eq:cv}) separates coordinated operations from coincidence.
High-degree funders representing exchange and bridge infrastructure are removed
before grouping to prevent spurious families.}
\label{fig:sybil}
\end{figure}

\subsubsection{Graduated penalty structure}

The four layers feed a penalty multiplier anchored to four qualitative tiers.
Wallets in the lowest tier receive full attribution. Wallets exhibiting
ambiguous behavioral signals receive no penalty. Penalties activate only when
multiple strong signals fire jointly, producing approximately a $50\%$ reduction
in the third tier and approximately $80\%$ in the fourth. Numeric breakpoints
are withheld. \Cref{fig:penalty} contrasts the resulting response with binary
exclusion: a legitimate wallet misclassified by one tier receives a proportional
reduction rather than losing everything.

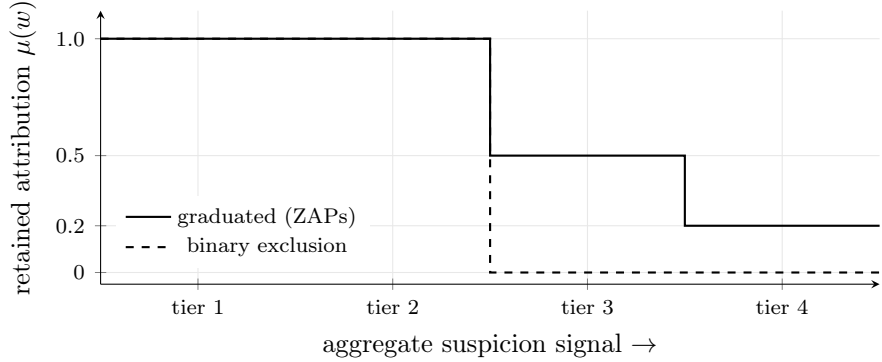
\begin{figure}[t]
\centering
\begin{tikzpicture}
\begin{axis}[
  width=0.72\textwidth, height=5.2cm,
  xlabel={aggregate suspicion signal $\rightarrow$},
  ylabel={retained attribution $\mu(w)$},
  xmin=0, xmax=4, ymin=-0.05, ymax=1.12,
  xtick={0.5,1.5,2.5,3.5},
  xticklabels={tier 1, tier 2, tier 3, tier 4},
  ytick={0,0.2,0.5,1.0},
  yticklabels={$0$,$0.2$,$0.5$,$1.0$},
  tick label style={font=\scriptsize}, label style={font=\small},
  legend style={font=\scriptsize, at={(0.02,0.06)}, anchor=south west, draw=none, fill=white},
  grid=major, grid style={gray!18},
  axis lines=left,
]
\addplot[thick, const plot, mark=none, black]
  coordinates {(0,1.0) (1,1.0) (2,0.5) (3,0.2) (4,0.2)};
\addlegendentry{graduated (ZAPs)}
\addplot[thick, const plot, mark=none, black, dashed]
  coordinates {(0,1.0) (1,1.0) (2,0.0) (3,0.0) (4,0.0)};
\addlegendentry{binary exclusion}
\end{axis}
\end{tikzpicture}
\caption{Graduated versus binary penalty response. Under binary exclusion a
one-tier misclassification is total loss; the graduated curve bounds
false-positive harm to the tier increment.}
\label{fig:penalty}
\end{figure}

\section{Results}
\label{sec:results}

\subsection{Reference Deployment Statistics}

\Cref{tab:deployment} summarizes distributional properties of the reference
deployment; \Cref{tab:percentiles} reports the reward distribution across
percentiles.

\begin{table}[t]
\centering
\caption{Reference deployment summary.}
\label{tab:deployment}
\begin{tabular}{lr}
\toprule
Metric & Value \\
\midrule

Wallets scored & 320M+ \\
Aggregate volume & \$300B \\
Total reward units distributed & 16{,}080 \\
Gini coefficient & 0.718 \\
Top-10-wallet reward share & 0.4\% \\
Top-100-wallet reward share & 2.8\% \\
Distinct score value ratio & 63.4\% \\
P99 / P50 reward ratio & 36.6$\times$ \\
\bottomrule
\end{tabular}
\end{table}

\begin{table}[t]
\centering
\caption{Reward distribution percentiles.}
\label{tab:percentiles}
\begin{tabular}{lrrrrrrrr}
\toprule
Percentile & P10 & P25 & P50 & P75 & P90 & P95 & P99 & Max \\
\midrule
Reward & 0.003 & 0.010 & 0.084 & 0.689 & 1.843 & 2.307 & 3.089 & 6.597 \\
\bottomrule
\end{tabular}
\end{table}

Each headline statistic verifies a distinct design property. The Gini
coefficient of $0.718$ sits inside the band treated as healthy for incentive
distributions: below it, rewards fail to differentiate contribution levels;
above it, rewards collapse into concentration. Perfect equality would itself be
a design failure, since the objective is proportionate inequality rather than
flatness.

Top-decile concentration is the direct read-out of the percentile cap in
\Cref{eq:volume}. The ten largest recipients hold $0.4\%$ of all rewards and the
top hundred hold $2.8\%$---figures that would be one to two orders of magnitude
higher under uncapped volume-proportional attribution on a population of this
size. The distinct-score-value ratio of $63.4\%$ confirms a fine-grained
continuous gradient rather than bucket collapse; a bracketed scheme over the same
population would produce a ratio near zero. The P99/P50 spread of $36.6\times$
quantifies what the cap permits: a top-percentile participant out-earns the
median by roughly $37\times$, enough that genuine contribution is visibly worth
making, bounded enough that the median reward remains non-trivial.

\subsection{Sector-Level Distribution}

\Cref{tab:sectors} reports volume, sector weighting, and reward share across the
five sectors.

\begin{table}[t]
\centering
\caption{Sector-level volume, weighting, and reward shares in the reference
deployment.}
\label{tab:sectors}
\begin{tabular}{lrrr}
\toprule
Sector & Volume & Sector weight $V_\sigma/V_{\Ecos}$ & Reward share \\
\midrule
DEX swap    & \$35.3B & 0.8570 & 93.8\% \\
Lending     & \$4.2B  & 0.1015 & 2.8\%  \\
Perpetuals  & \$1.1B  & 0.0279 & 1.8\%  \\
NFT         & \$369M  & 0.0090 & 0.7\%  \\
Staking     & \$191M  & 0.0047 & 0.9\%  \\
\bottomrule
\end{tabular}
\end{table}

Reward shares track volume shares closely by construction---the property
\Cref{eq:weight} exists to enforce. One deployment observation illustrates its
bite. The protocol with the largest wallet count in the dataset, over half of all
scored wallets, received $1.3\%$ of total rewards, because its volume share
within its sector was $0.6\%$. Under wallet-count or transaction-count
attribution the same protocol would have dominated the distribution outright.
Head-count is the cheapest metric to manufacture, and \Cref{prop:niche} prices
it accordingly.

\subsection{Detector Evaluation on Labeled Malicious Wallets}
\label{sec:detector-eval}

We evaluate the ensemble of \Cref{sec:ensemble} on a corpus of $1{,}073$
labeled malicious wallets spanning $124{,}638$ transactions, drawn from public
exploit and phishing attributions, against a benign comparison set of $383$
wallets spanning $326{,}943$ transactions. The reconstruction model is retrained
one-class on a $70\%$ fold of the malicious corpus; the held-out $30\%$
($322$ wallets) plus the full benign set form the evaluation set, so no
evaluated malicious wallet was seen in training. All figures are means over five
random seeds. \Cref{tab:detectors} reports the result.

\begin{table}[t]
\centering
\caption{Detector performance on held-out malicious wallets. The isolation
forest is reported under both training regimes. ROC-AUC, mean $\pm$ s.d.\ over
five seeds; $322$ held-out malicious and $383$ benign wallets.}
\label{tab:detectors}
\begin{tabular}{llcc}
\toprule
Detector & Training regime & ROC-AUC & $\mathrm{corr}(P_{\mathrm{A}},P_{\mathrm{B}})$ \\
\midrule
Reconstruction model $P_{\mathrm{A}}$ & one-class, malicious only & $0.891 \pm 0.016$ & --- \\
Isolation forest $P_{\mathrm{B}}$     & pooled population         & $0.250 \pm 0.015$ & $-0.686$ \\
Isolation forest $P_{\mathrm{B}}$     & benign population         & $0.638 \pm 0.011$ & $-0.217$ \\
\midrule
Ensemble, $\lambda=1.00$              & (equivalent to $P_{\mathrm{A}}$ alone) & $0.891 \pm 0.016$ & --- \\
Ensemble, $\lambda^\ast$, pooled fit  & best over sweep           & $0.891 \pm 0.016$ & --- \\
\textbf{Ensemble, $\lambda^\ast=0.90$, benign fit} & best over sweep & $\mathbf{0.923 \pm 0.013}$ & --- \\
\bottomrule
\end{tabular}
\end{table}

Three findings follow. First, the one-class reconstruction model is the stronger
individual detector at $0.891$ ROC-AUC, consistent with its access to labeled
attack archetypes. Second, the isolation forest's training regime is
outcome-determining. Fit on the pooled population it achieves $0.250$---far
below chance, meaning statistical outlier-ness \emph{inversely} predicts
maliciousness. The malicious wallets outnumber the benign ones in the pooled
corpus and are individually sparser (median $21$ transactions against $499$), so
they define the dense region of the feature space and the genuinely
heavy-activity benign wallets are isolated as outliers. Fit on the benign
population alone, as a proper normality model, the same detector reaches
$0.638$ with correct polarity. Third, only the benign-fit configuration produces
an ensemble that beats its strongest member: $0.923$ against $0.891$, a gain of
$+0.032$ ROC-AUC at $\lambda^\ast = 0.90$. Under the pooled fit the optimum sits
at $\lambda = 1.00$, which discards the second detector entirely.

\Cref{fig:sweep} plots the full sweep. The low correlation between detectors
under the benign fit ($-0.217$) is what makes the ensemble gain available; the
strong negative correlation under the pooled fit ($-0.686$) reflects the
polarity inversion rather than complementarity.

\begin{figure}[t]
\centering
\begin{tikzpicture}
\begin{axis}[
  width=0.78\textwidth, height=6.2cm,
  xlabel={fusion weight $\lambda$ \ \ (weight on reconstruction model $P_{\mathrm{A}}$)},
  ylabel={ROC-AUC},
  xmin=0, xmax=1, ymin=0.18, ymax=0.99,
  xtick={0,0.2,0.4,0.6,0.8,1.0},
  ytick={0.2,0.4,0.6,0.8,1.0},
  tick label style={font=\scriptsize}, label style={font=\small},
  legend style={font=\scriptsize, at={(0.98,0.03)}, anchor=south east,
                draw=gray!40, fill=white, row sep=1pt},
  grid=major, grid style={gray!18},
]
\addplot[thick, black, mark=*, mark size=1.1pt] coordinates {
(0.00,0.6380)(0.05,0.6647)(0.10,0.6897)(0.15,0.7136)(0.20,0.7366)
(0.25,0.7624)(0.30,0.7913)(0.35,0.8193)(0.40,0.8459)(0.45,0.8683)
(0.50,0.8861)(0.55,0.8996)(0.60,0.9051)(0.65,0.9096)(0.70,0.9135)
(0.75,0.9156)(0.80,0.9190)(0.85,0.9219)(0.90,0.9233)(0.95,0.9224)
(1.00,0.8913)};
\addlegendentry{isolation forest fit on benign population}
\addplot[thick, black, dashed, mark=square, mark size=1.1pt] coordinates {
(0.00,0.2504)(0.05,0.2580)(0.10,0.2677)(0.15,0.2807)(0.20,0.3002)
(0.25,0.3419)(0.30,0.4049)(0.35,0.4386)(0.40,0.4760)(0.45,0.5220)
(0.50,0.5915)(0.55,0.6759)(0.60,0.7123)(0.65,0.7418)(0.70,0.7664)
(0.75,0.7910)(0.80,0.8097)(0.85,0.8248)(0.90,0.8358)(0.95,0.8658)
(1.00,0.8913)};
\addlegendentry{isolation forest fit on pooled population}
\addplot[gray, dotted, thick] coordinates {(0,0.8913)(1,0.8913)};
\addlegendentry{$P_{\mathrm{A}}$ alone}
\end{axis}
\end{tikzpicture}
\caption{Ensemble ROC-AUC across the fusion weight $\lambda$ of
\Cref{eq:fusion}. Only the benign-fit configuration rises above the
single-detector baseline, peaking at $\lambda=0.90$. Under the pooled fit the
sweep is monotone toward $\lambda=1$, so the second detector contributes
nothing.}
\label{fig:sweep}
\end{figure}
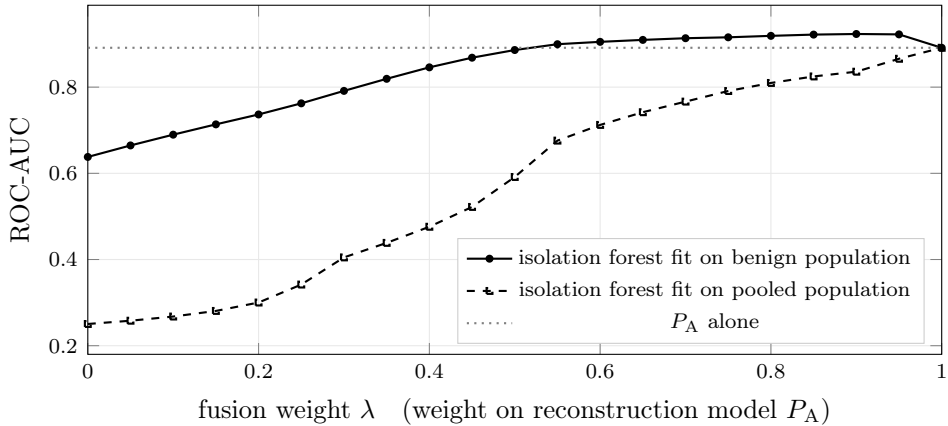

\subsection{Adversarial Robustness}

The detection and penalty stack is evaluated on eight constructed scenarios:
four adversarial archetypes drawn from observed farming patterns and four
legitimate-user controls. Each is scored through the full pipeline with and
without the adversarial defenses activated. \Cref{tab:adversarial} reports
absolute attribution and the resulting change.

\begin{table}[t]
\centering
\caption{Reward capture under adversarial and control scenarios, before and
after activation of the four-layer defense stack.}
\label{tab:adversarial}
\begin{tabular}{lrrr}
\toprule
Scenario & Undefended & Defended & Change \\
\midrule
\multicolumn{4}{l}{\emph{Adversarial archetypes}} \\
Diversity farmer (spread across spam protocols) & 10.53 & 1.11 & $-89\%$ \\
Spam bot (high-frequency, low-value, uniform) & 15.96 & 1.62 & $-90\%$ \\
Flash-loan exploiter (single-block round-trip) & 18.56 & 10.65 & $-43\%$ \\
Sybil operator (coordinated small-volume wallets) & 211.50 & 148.80 & $-30\%$ \\
\midrule
\multicolumn{4}{l}{\emph{Legitimate controls}} \\
Diversified multi-protocol holder & 52.38 & 53.11 & $+1\%$ \\
High-volume trader & 88.31 & 87.49 & $-1\%$ \\
Derivatives trader & 36.41 & 34.21 & $-6\%$ \\
NFT collector & 25.01 & 23.04 & $-8\%$ \\
\bottomrule
\end{tabular}
\end{table}

The headline property is asymmetry. Every adversarial archetype loses
double-digit percentages of reward capture, the two purest automation patterns
losing approximately $90\%$, while no legitimate control shifts by more than
$8\%$ and the most common legitimate profile is effectively untouched. This
asymmetry, rather than raw detection rate, is the correct figure of merit for an
attribution system: a mitigation layer that halved bot rewards while halving
power-user rewards would be a failure regardless of its recall.

The gradient across adversarial scenarios carries structural information. The
spam bot and diversity farmer are single-wallet strategies whose behavioral
signatures---cadence without idle periods, value uniformity, protocol
concentration---fall squarely within what the ensemble detects. The flash-loan
exploiter retains more ($-43\%$) because its distinguishing behavior
concentrates in a small number of high-value events rather than a sustained
pattern. The sybil operator retains the most ($-30\%$): each constituent wallet
viewed in isolation resembles a modest legitimate participant, and only the
clustering layer of \Cref{sec:sybil} reaches it.

\subsection{Sybil Clustering Validation}

A validation scan over $500$ wallets resolved $32$ candidate families spanning
$97$ wallets, after excluding $39$ high-degree first funders corresponding to
exchange and bridge infrastructure. Nine families passed the four-metric funder
validation gate. \Cref{tab:sybil} reports the family-size distribution and
associated risk scores.

\begin{table}[t]
\centering
\caption{Sybil families resolved in a 500-wallet validation scan. Risk score
composes family size, funder validation, and the uniformity statistic of
\Cref{eq:cv}.}
\label{tab:sybil}
\begin{tabular}{rrrr}
\toprule
Family size & Families & Wallets & Risk-score range \\
\midrule
2  & 17 & 34 & 0.39 -- 0.57 \\
3  & 7  & 21 & 0.42 -- 0.60 \\
4  & 4  & 16 & 0.42 -- 0.52 \\
5  & 1  & 5  & 0.63 \\
6  & 2  & 12 & 0.45 -- 0.55 \\
8  & 1  & 8  & 0.45 \\
\midrule
\textbf{Total} & \textbf{32} & \textbf{97} & 0.39 -- 0.63 \\
\bottomrule
\end{tabular}
\end{table}

Family sizes skew small: $17$ of $32$ families are wallet pairs. The largest
resolved family spans eight wallets. Risk scores concentrate in a narrow band
($0.39$--$0.63$), consistent with families that share funding provenance and
exhibit partial behavioral uniformity rather than mechanical duplication. One
two-wallet family moved approximately $\$191$M inbound against $\$191$M outbound
with a low DeFi-activity share, a near-perfect pass-through signature more
consistent with exchange infrastructure than with a farming operation; cases of
this shape are the reason the funder-exclusion step precedes grouping.

\subsection{Bot Classification Population Breakdown}

\Cref{tab:tiers} reports the distribution of wallets across the four graduated
penalty tiers.

\begin{table}[t]
\centering
\caption{Classification tiers over the reference population.}
\label{tab:tiers}
\begin{tabular}{lrrrr}
\toprule
Tier & Wallets & Share & Penalty & Aggregate rewards \\
\midrule
Normal     & 10{,}366 & 33\% & none        & 465 \\
Suspicious & 13{,}819 & 44\% & none        & 13{,}547 \\
Likely bot & 7{,}532  & 24\% & $\sim$50\%  & 2{,}067 \\
Hard bot   & 0        & 0\%  & $\sim$80\%  & --- \\
\bottomrule
\end{tabular}
\end{table}

Two observations anchor the interpretation. First, the Suspicious tier, flagged
but unpenalized, carries $84\%$ of all distributed rewards. This is the
no-penalty design decision doing its intended work. The wallets
generating the most substantial genuine activity are precisely those whose
operational regularity trips partial behavioral signals: power users, market
makers, and institutional accounts whose disciplined patterns resemble
automation. A penalty policy triggered on partial evidence would have
confiscated the bulk of legitimately earned rewards.

Second, the empty Hard Bot tier reflects the reference population rather than a
gap in the mechanism. No wallet fired extreme signals across every detector
simultaneously. The tier exists as the asymptotic bound the graduated penalty
curve applies when one does.

\subsection{Production Campaign Results}

In live token-distribution campaigns deploying the full stack against real
adversarial populations, the framework recorded a $56\%$ reduction in sybil
allocation, a $49\%$ increase in quality-wallet participation, and a $50\%$
reduction in post-distribution sell pressure, measured against comparable
unmitigated campaigns. Allocation to low-quality wallets fell from $21.0\%$ to
$13.5\%$ of the distributed pool while allocation to high-quality wallets rose
from $41.7\%$ to $54.7\%$, and mean post-distribution holding duration extended
from $9.8$ to $12.6$ days.

\subsection{Computational Cost}

Scoring is linear in transaction count per wallet and parallelizes across
wallets without coordination, a property inherited from the deterministic
hashing of \Cref{sec:behavioral}. A wallet carrying $40{,}000$ transactions
scores in $0.295$\,s. Across a $1{,}400$-wallet synthetic population, total
runtime was $2.48$\,s at a mean of $1.77$\,ms and a P95 of $3.19$\,ms per
wallet.

\section{Ablation Studies}
\label{sec:ablation}

\subsection{Normalization Reference}

\Cref{eq:volume} anchors the volume score to a high population percentile. Two
alternatives fail in opposite directions. Under max-normalization the reference
becomes a function of the single largest participant, so a lone anomalous
wallet---or a pricing artifact surviving the ingestion ceiling---compresses
every other score toward zero and the distinct-score ratio collapses. Under
mean-normalization the reference sits inside the bulk of a heavy-tailed
distribution, so the clip binds for a large fraction of participants, scores
pile at the cap, and differentiation across the majority is lost. The percentile
anchor is the configuration for which neither degeneracy occurs: robust to
individual outliers, far enough into the tail that saturation reaches only
genuine whales.

\subsection{Ensemble Composition}

\Cref{tab:detectors} and \Cref{fig:sweep} constitute a leave-one-detector-out ablation over
\Cref{eq:fusion}, since $\lambda=1$ and $\lambda=0$ recover the individual
detectors. Removing the isolation forest ($\lambda=1$) costs $0.032$ ROC-AUC
under the benign-fit configuration. Removing the reconstruction model
($\lambda=0$) costs $0.285$. The reconstruction model is the load-bearing
component; the isolation forest is a genuine but secondary contributor, and only
under the correct training regime.

The training-regime result deserves emphasis as a finding in its own right.
Isolation forests are routinely described as unsupervised detectors that may be
fit on whatever data is available. In this setting that framing is actively
harmful: pooled fitting produced a detector at $0.250$ ROC-AUC, which is worse
than the trivial constant predictor and would degrade any ensemble it entered.
The corrective condition is that the forest must model the \emph{benign}
population specifically, so that outlier-ness and maliciousness point in the
same direction. Any deployment that assembles its unlabeled training pool from
production traffic---which necessarily contains the attackers---inherits this
inversion in proportion to the attacker prevalence.

\subsection{Penalty Shape}

Replacing the graduated multiplier with binary exclusion at the third tier
(\Cref{fig:penalty}) would remove $2{,}067$ reward units of Likely Bot
attribution instead of approximately half that. The relevant cost is borne by
false positives. Under the deployment distribution of \Cref{tab:tiers}, a
legitimate wallet misclassified into the third tier retains $50\%$ of its
attribution under the graduated curve and $0\%$ under exclusion. Given that the
Suspicious tier---adjacent to Likely Bot and explicitly unpenalized---carries
$84\%$ of all distributed rewards, the population most exposed to a one-tier
misclassification is also the population holding the majority of legitimate
value.

\section{Limitations and Future Work}
\label{sec:limitations}

\paragraph{Quality dimension is static in the reference deployment.} The
$q_{w,p}$ term of \Cref{eq:composite} is held at a calibrated baseline for all
users in the deployment reported here, so \Cref{eq:composite} effectively
reduces to a two-signal composite plus a constant floor. Sector-specific quality
signals---repayment discipline, realized-PnL character, token-quality mix---are
specified in \Cref{sec:behavioral} and integrated in the behavioral engine, but
the reported distributional statistics do not yet reflect them.

\paragraph{Benign comparison set is small and activity-skewed.} The evaluation
of \Cref{sec:detector-eval} uses $383$ benign wallets with a median of $499$
transactions, against malicious wallets with a median of $21$. The benign sample
therefore over-represents high-activity participants. This skew is precisely
what drives the polarity inversion under pooled fitting, and while the direction
of the effect is robust across seeds and forest sizes, its magnitude should be
expected to depend on the benign sampling frame. A representative benign panel
would sharpen the estimate of the isolation forest's standalone contribution.

\paragraph{Labeled malicious corpus reflects exploit attribution, not farming.}
The $1{,}073$-wallet corpus is drawn from public exploit and phishing
attributions. Airdrop farming and exploit activity are related but distinct
behavioral classes, so \Cref{tab:detectors} measures transfer to a proxy
population rather than direct farming detection. The scenario results of
\Cref{tab:adversarial} address farming archetypes directly but are constructed
rather than observed.

\paragraph{Post-distribution memory is unquantified.} The mechanism of
Layer 3 is specified qualitatively. Neither the liquidation-share threshold, the
observation window, nor the decay of the eligibility adjustment across cycles is
characterized here, and no ablation isolates its contribution to the production
results.

\paragraph{Cross-chain identity.} The clustering layer of \Cref{sec:sybil}
operates on single-chain funding provenance. An operation that funds wallets
across chains, or through bridges whose deposit addresses are excluded as
high-degree funders, evades first-funder grouping. Cross-chain identity
resolution is the most consequential open extension.

\paragraph{Detector drift.} Both detectors are fit offline. An adaptive
adversary can probe the deployed penalty structure and migrate behavior toward
regions that reconstruct poorly and appear statistically ordinary. Online
retraining against observed behavioral drift, and periodic recalibration of
$\lambda$, are required for the asymmetry of \Cref{tab:adversarial} to persist.

\section{Conclusion}
\label{sec:conclusion}

We presented ZAPs, a reward-attribution framework for DeFi ecosystems that
treats adversarial robustness as a design requirement rather than a retrofit.
The framework composes a percentile-normalized activity score that bounds whale
dominance while preserving differentiation; a nested cross-domain weighting
mechanism for which \Cref{prop:niche} establishes that extractable attribution
at any protocol is bounded by that protocol's global volume share, closing the
niche-protocol exploit structurally rather than heuristically; a four-layer
adversarial detection stack feeding graduated penalties that bound
false-positive harm; and an attribution algebra composing these into a single
per-wallet reward.

On a held-out evaluation over $1{,}073$ labeled malicious wallets, the parallel
anomaly ensemble reaches $0.923 \pm 0.013$ ROC-AUC against $0.891 \pm 0.016$ for
the reconstruction model alone. That gain is conditional on a training-regime
requirement we believe has not been reported: an isolation forest fit on the
pooled population rather than the benign population inverts its discriminative
polarity, scoring at $0.250$ and contributing nothing to any ensemble. Under
adversarial simulation the mitigation stack removes $30$--$90\%$ of attacker
reward capture at a cost of $1$--$8\%$ to legitimate-user scenarios. Live
production campaigns recorded a $56\%$ reduction in sybil allocation, a $49\%$
increase in quality-wallet participation, and a $50\%$ reduction in sell
pressure.

Reward attribution, airdrop eligibility, sybil-resistant incentive design, and
composable on-chain reputation reduce to the same underlying question---what did
this wallet actually contribute?---and the framework described here is one
deployed, measured answer to it.


\end{document}